\documentclass{article}
\usepackage[top=1in, bottom=1in, left=1.25in, right=1.25in]{geometry}
\usepackage{hyperref}
\usepackage{amsmath, amsthm, amssymb,algorithmic,algorithm,graphicx,subfigure}
\usepackage{amsfonts}
\usepackage{tikz}
\usetikzlibrary{plotmarks}

\newtheorem{theorem}{\textbf{Theorem}}[section]
\newtheorem{definition}[theorem]{\textbf{Definition}}
\newtheorem{corollary}[theorem]{\textbf{Corollary}}
\newtheorem{lemma}[theorem]{\textbf{Lemma}}

\newtheorem{conjecture}[theorem]{\textbf{Conjecture}}

\begin{document}

\title{The Bethe Partition Function of Log-supermodular Graphical Models}
\author{Nicholas Ruozzi}
\author{
Nicholas Ruozzi\footnote{This work was supported by EC grant FP7-265496, ``STAMINA''.}\\
Communication Theory Laboratory\\
EPFL\\
Lausanne, Switzerland \\
\texttt{nicholas.ruozzi@epfl.ch} \\
}
\date{}
\maketitle

\begin{abstract}
Sudderth, Wainwright, and Willsky have conjectured that the Bethe approximation corresponding to any fixed point of the belief propagation algorithm over an attractive, pairwise binary graphical model provides a lower bound on the true partition function.  In this work, we resolve this conjecture in the affirmative by demonstrating that, for any graphical model with binary variables whose potential functions (not necessarily pairwise) are all log-supermodular, the Bethe partition function always lower bounds the true partition function.  The proof of this result follows from a new variant of the ``four functions'' theorem that may be of independent interest.
\end{abstract}

\section{Introduction}

Graphical models have proven to be a useful tool for performing approximate inference in a wide variety of application areas including computer vision, combinatorial optimization, statistical physics, and wireless networking. Computing the partition function of a given graphical model,  a typical inference problem, is  an NP-hard problem in general.  Because of this, the inference problem is often replaced by a variational approximation that is, hopefully, easier to solve.  The Bethe approximation, one such standard approximation, is of great interest both because of its practical performance and because of its relationship to the belief propagation (BP) algorithm:  stationary points of the Bethe free energy function correspond to fixed points of belief propagation \cite{yedweiss}. However, the Bethe partition function is only an approximation to the true partition function and need not provide an upper or lower bound.   

In certain special cases, the Bethe approximation is conjectured to provide a bound on the true partition function.  One such example is the class of attractive pairwise graphical models: models in which the interaction between any two neighboring variables places a greater weight on assignments in which the two variables agree.   Many applications in computer vision and statistical physics can be expressed as attractive pairwise graphical models (e.g., the ferromagnetic Ising model). Sudderth, Wainwright, and Willsky \cite{sudderth07} used a loop series expansion of Chertkov and Chernyak \cite{chert, gomez} in order to study the fixed points of BP over attractive graphical models. They provided conditions on the fixed points of BP under which the stationary points of the Bethe free energy function corresponding to these fixed points is a lower bound on the true partition function.  Empirically, they observed that, even when their conditions were not satisfied, the Bethe partition function appeared to lower bound the true partition function, and they conjectured that this is always the case for attractive pairwise binary graphical models.

Recent work on the relationship between the Bethe partition function and the graph covers of a given graphical model has suggested a new approach to resolving this conjecture.  Vontobel \cite{pascalcounting} demonstrated that the Bethe partition function can be precisely characterized by the average of the true partition functions corresponding to covers of the base graphical model.  The primary contribution of the present work is to show that, for graphical models with log-supermodular potentials, the partition function associated with any graph cover of the base graph, appropriately normalized, must lower bound the true partition function.  As pairwise binary graphical models are log-supermodular if and only if they are attractive, combining our result with the observations of \cite{pascalcounting} resolves the conjecture of \cite{sudderth07}.  

The key element in our proof, and the second contribution of this work, is a new variant of the ``four functions'' theorem that is specific to log-supermodular functions.  We state and prove this variant in Section \ref{sec:kfunc}, and in Section \ref{sec:part}, we use it to resolve the conjecture. As a final contribution, we demonstrate that our variant of the ``four functions'' theorem has applications beyond log-supermodular functions: we use it to show that the Bethe partition function can also provide a lower bound on the number of independent sets in a bipartite graph.

\section{Undirected Graphical Models}
Let $f:\{0,1\}^n\rightarrow \mathbb{R}_{\geq 0}$ be a non-negative function.  We say that $f$ factors with respect to a hypergraph $G = (V, \mathcal{A})$ where $\mathcal{A} \subseteq 2^{V}$, if there exist potential functions $\phi_i:\{0,1\}\rightarrow \mathbb{R}_{\geq 0}$ for each $i\in V$ and $\psi_{\alpha}:\{0,1\}^{|\alpha|}\rightarrow \mathbb{R}_{\geq 0}$ for each $\alpha\in \mathcal{A}$ such that
\[f(x) = \prod_{i\in V} \phi_i(x_i) \prod_{\alpha\in \mathcal{A}} \psi_{\alpha}(x_\alpha)\]
where $x_\alpha$ is the subvector of the vector $x$ indexed by the set $\alpha$.

We will express the hypergraph $G$ as a bipartite graph that consists of a variable node for each $i\in V$, a factor node for each $\alpha\in\mathcal{A}$, and an edge joining the factor node corresponding to $\alpha$ to the variable node
representing $i$ if $i \in \alpha$.  This is typically referred to as the \textit{factor graph} representation of $G$. 

\begin{definition}
A function $f:\{0,1\}^n\rightarrow \mathbb{R}_{\geq 0}$ is \textbf{log-supermodular} if for all $x,y\in\{0,1\}^n$
\[f(x)f(y) \leq f(x\wedge y) f(x\vee y)\]
where $(x\wedge y)_i = \min \{x_i, y_i\}$ and $(x\vee y)_i = \max \{x_i, y_i\}$.
Similarly, a function $f:\{0,1\}^n\rightarrow \mathbb{R}_{\geq 0}$ is \textbf{log-submodular} if for all $x,y\in\{0,1\}^n$
\[f(x)f(y) \geq f(x\wedge y) f(x\vee y)\]
\end{definition}

\begin{definition}
A factorization of a function $f:\{0,1\}^n\rightarrow \mathbb{R}_{\geq 0}$ over $G=(V,\mathcal{A})$ is log-supermodular if for all $\alpha\in\mathcal{A}$, $\psi_\alpha(x_\alpha)$ is log-supermodular.  
\end{definition}

Every function that admits a log-supermodular factorization is necessarily log-supermodular as products of log-supermodular functions are easily seen to be log-supermodular, but the converse may not be true outside of special cases.  If $|\alpha| \leq 2$ for each $\alpha\in\mathcal{A}$, then we call the factorization pairwise.  For any pairwise factorization, $f$ is log-supermodular if and only if $\psi_{ij}$ is log-supermodular for each $i$ and $j$.  

Pairwise graphical models such that $\psi_\alpha(x_\alpha)$ is log-supermodular for all $\alpha\in\mathcal{A}$ are referred to as \textit{attractive} graphical models. A generalization of attractive interactions to the non-pairwise case is presented in \cite{sudderth07}: for all $\alpha\in\mathcal{A}$, $\psi_\alpha$, when appropriately normalized, has non-negative central moments.  

\subsection{Graph Covers}
Graph covers have played an important role in our understanding graphical models \cite{pascalcounting, vontobel}.  

\begin{definition}
A graph $H$ \textbf{covers} a graph $G = (V, E)$ if there exists a graph homomorphism
$h: H \rightarrow G$ such that for all vertices $v\in G$ and all $w\in h^{-1}(v)$, $h$ maps the neighborhood
$\partial w$ of $w$ in $H$ bijectively to the neighborhood $\partial v$ of $v$ in $G$.  If $h(w) = v$, then we say that $w\in
H$ is a copy of $v\in G$.  Further, $H$ is a $k$-cover of $G$ if every vertex of
$G$ has exactly $k$ copies in $H$.
\end{definition}

Roughly, if a graph $H$ covers a graph $G$, then $H$ looks locally the same as $G$.  For an example of a graph cover, see Figure \ref{fig:cover}.

\begin{figure}
\centering
\subfigure[][A graph, $G$.]{
  \begin{tikzpicture}[scale=1.5]
	\tikzstyle{every node}=[draw,shape=circle];
	\path (0,0) node (X0) {$1$};
	\path (1,0) node (X1) {$2$};
	\path (1,-1) node (X2) {$3$};
	\path (0,-1) node (X3) {$4$};
	
	\draw (X0) -- (X1);
	\draw (X0) -- (X2);
	\draw (X2) -- (X3);
	\draw (X1) -- (X2);
	\draw (X1) -- (X3);
	\end{tikzpicture}
}\hspace{1cm}
\subfigure[][One possible cover of $G$.]{  
	\begin{tikzpicture}[scale=3]
	\tikzstyle{every node}=[draw,shape=circle];
	\path (0,0) node (X0) {$1$};
	\path (.5,0) node (X1) {$2$};
	\path (1,0) node (X2) {$3$};
	\path (1.5,0) node (X3) {$4$};
	\path (0,-.5) node (Y0) {$1$};
	\path (.5,-.5) node (Y1) {$2$};
	\path (1,-.5) node (Y2) {$3$};
	\path (1.5,-.5) node (Y3) {$4$};
	
	\draw (X0) -- (X1);
	\draw (X0) -- (Y2);
	\draw (Y1) -- (Y0);
	\draw (Y1) -- (Y2);
	\draw (X1) -- (Y3);
	\draw (X2) -- (Y0);
	\draw (X2) -- (X1);
	\draw (X2) -- (Y3);
	\draw (X3) -- (Y1);
	\draw (X3) -- (Y2);
		
	\end{tikzpicture}}
\caption{An example of a graph cover.  The nodes in the cover are labeled for the node that they copy in the base graph.}
\label{fig:cover}
\end{figure}
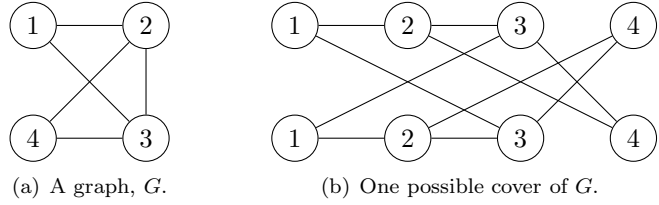

For the factor graph corresponding to $G = (V,\mathcal{A})$, each $k$-cover consists of a variable node for each of the $k|V|$ variables, a factor node for each of the $k|\mathcal{A}|$ factors, and an edge joining each copy of $\alpha\in\mathcal{A}$ to a distinct copy of each $i\in\alpha$.  To any $k$-cover $H = (V_H, \mathcal{A}_H)$ of $G$, we can associate a collection of potentials: the
potential at node $i\in V_H$ is equal to $\phi_{h(i)}$, the potential at node $h(i) \in G$, and for each $\alpha\in\mathcal{A}_H$, we associate the potential $\psi_{h(\alpha)}$.  In this way, we can construct a function $f^H:\{0,1\}^{kn}\rightarrow \mathbb{R}_{\geq 0}$ such that $f^H$ factorizes over $H$.

Notice that if $f^G$ admits a log-supermodular factorization over $G$ and $H$ is a $k$-cover of $G$, then $f^H$ admits a log-supermodular factorization over $H$.

\subsection{Bethe Approximations}
For a function $f:\{0,1\}^n\rightarrow \mathbb{R}_{\geq 0}$ that factorizes over $G = (V,\mathcal{A})$, we are interested computing the partition function $Z(G) = \sum_x f(x)$.  In general, this is an NP-hard problem, but in practice, algorithms, such as belief propagation, based on variational approximations produce reasonable estimates in certain settings.  One such variational approximation, the Bethe approximation at temperature $T = 1$, is defined as follows:
\begin{align*}
\log Z_\mathrm{B}(G, \tau) &= \sum_{i\in V}\sum_{x_i} \tau_i(x_i)\log\phi_i(x_i) + \sum_{\alpha\in\mathcal{A}} \sum_{x_\alpha} \tau_\alpha(x_\alpha)\log\psi_\alpha(x_\alpha)\\
&\:{-} \sum_{i\in V}\sum_{x_i} \tau_i(x_i)\log\tau_i(x_i) - \sum_{\alpha\in\mathcal{A}} \sum_{x_\alpha} \tau_\alpha(x_\alpha)\log \frac{\tau_\alpha(x_\alpha)}{\prod_{i\in\alpha} \tau_i(x_i) }
\end{align*}
for $\tau$ in the local marginal polytope, \[\mathcal{T} \triangleq \{\tau\geq 0\hspace{.1cm} |\hspace{.1cm} \forall \alpha\in\mathcal{A}, i\in\alpha, \sum_{x_{\alpha\setminus i}} \tau_\alpha(x_\alpha) = \tau_i(x_i)\text{ and } \forall i\in V, \sum_{x_i} \tau_i(x_i) = 1\}.\]

The fixed points of the belief propagation algorithm correspond to stationary points of $\log Z_{\mathrm{B}}(G, \tau)$ over $\mathcal{T}$, the set of pseudomarginals  \cite{yedweiss}, and the Bethe partition function is defined to be the maximum value achieved by this approximation over $\mathcal{T}$:
\[Z_\mathrm{B}(G) = \max_{\tau\in\mathcal{T}} Z_\mathrm{B}(G, \tau).\]

For a fixed factor graph $G$, we are interested in the relationship between the true partition function, $Z(G)$, and the Bethe approximation corresponding to $G$, $Z_\mathrm{B}(G)$.  While, in general, $Z_\mathrm{B}(G)$ can be either an upper or a lower bound on the true partition function, in this work, we address the following conjecture of \cite{sudderth07}:
\begin{conjecture}
If $f:\{0,1\}^n\rightarrow \mathbb{R}_{\geq 0}$ admits a pairwise, log-supermodular factorization over $G = (V,\mathcal{A})$, then $Z_\mathrm{B}(G) \leq Z(G)$.\label{conj:bethecover}
\end{conjecture}

We resolve this conjecture in the affirmative, and show that it continues to hold for a larger class of log-supermodular functions.  Our results are based, primarily, on two observations:  a variant of the ``four functions'' theorem \cite{4func} and the following, recent, theorem of Vontobel \cite{pascalcounting}:

\begin{theorem}
\[Z_\mathrm{B}(G) = \lim \sup_{k\rightarrow\infty} \sqrt[k]{\sum_{H\in \mathcal{C}^k(G)} Z(H)/|\mathcal{C}^k(G)|}\]
where $\mathcal{C}^k(G)$ is the set of all $k$-covers of $G$. \label{thm:pascal}
\end{theorem}
\begin{proof}
See Theorem 27 of \cite{pascalcounting}.
\end{proof}

Theorem \ref{thm:pascal} suggests that a reasonable strategy for proving that $Z_\mathrm{B}(G) \leq Z(G)$ would be to show that $Z(H) \leq Z(G)^k$ for any $k$-cover $H$ of $G$.  This is the strategy that we adopt in the remainder of this work.  

\section{The ``Four Functions'' Theorem and Related Results}

Let $z^i$ be a function that computes the $i^{th}$ largest element of a collection.  We will, abusively, denote this function as $z^i(x^1,\ldots,x^k)$ for any collection of vectors $x^1,\ldots,x^k\in\mathbb{R}^n$.  Here, $z^i(x^1,\ldots,x^k)$ is the vector whose $j^{th}$ component is the $i^{th}$ largest element of $x^1_j,\ldots,x^k_j$ for each $j\in\{1,\ldots,n\}$.  As an example, for vectors $x^1,\ldots,x^k\in\{0,1\}^n$, $z^i(x^1,\ldots,x^k)_j = \{\sum_{a=1}^k x^a_j \geq i\}$ where $\{\cdot \geq \cdot\}$ is one if the inequality is satisfied and zero otherwise.  

The ``four functions'' theorem \cite{4func} is a general result concerning nonnegative functions over distributive lattices.  Many correlation inequalities, such as the FKG inequality, can be seen as special cases of this theorem \cite{alonspencer}.  
\begin{theorem}[``Four Functions'' Theorem]
Let $f_1,f_2,f_3,f_4:\{0,1\}^n\rightarrow \mathbb{R}_{\geq 0}$ be nonnegative real-valued functions.  If for all $x, y\in\{0,1\}^n$,
\[f_1(x)f_2(y) \leq f_3(x\wedge y)f_4(x\vee y),\]
then
\[ \Big[\sum_{z\in\{0,1\}^n} f_1(z)\Big]\Big[\sum_{z\in\{0,1\}^n} f_2(z)\Big] \leq \Big[\sum_{z\in\{0,1\}^n} f_3(z)\Big]\Big[\sum_{z\in\{0,1\}^n} f_4(z)\Big].\]
\end{theorem}

The following lemma is a direct consequence of the four functions theorem:
\begin{lemma}
If $f:\{0,1\}^n\rightarrow \mathbb{R}_{\geq 0}$ is log-supermodular, then every marginal of $f$ is also log-supermodular.\label{lem:logsup}
\end{lemma}

The four functions theorem can be generalized to more than four functions, and a special case of the more general ``2k functions'' theorem is as follows \cite{aharoni, rinott, rinott2}:
\begin{theorem}[``2k Functions'' Theorem]
Let $f_1,\ldots,f_k:\{0,1\}^n\rightarrow \mathbb{R}_{\geq 0}$ and $g_1,\ldots,g_k:\{0,1\}^{n}\rightarrow \mathbb{R}_{\geq 0}$ be nonnegative real-valued functions.  If for all $x^1,\ldots,x^k\in\{0,1\}^n$,
\begin{align}
\prod_{i=1}^k g_i(x^i) \leq \prod_{i=1}^k f_i(z^i(x^1,\ldots,x^k)),\label{eqn:2kfnc}
\end{align}
then
\[ \prod_{i=1}^k \Big[\sum_{x\in\{0,1\}^n} g_i(x)\Big] \leq \prod_{i=1}^k \Big[\sum_{x\in\{0,1\}^n} f_i(x)\Big].\]
\label{thm:rinott}
\end{theorem}

\subsection{A Variant of the ``Four Functions'' Theorem}
\label{sec:kfunc}
A natural generalization of Theorem \ref{thm:rinott} would be to replace the product of functions on the left-hand side of Equation \ref{eqn:2kfnc} with an arbitrary function over $x^1,\ldots,x^k$.  While the conclusion of the theorem may not continue to hold for arbitrary choices of such a function, we will show that we can replace this product with an arbitrary log-supermodular function while preserving the conclusion of the theorem.  The key property of log-supermodular functions that makes this possible is the following lemma:

\begin{lemma}
If $g:\{0,1\}^n\rightarrow \mathbb{R}_{\geq 0}$ is log-supermodular, then for any integer $k\geq 1$ and $x^1,\ldots,x^k\in\{0,1\}^n$,
\[\prod_{i=1}^k g(x^i) \leq \prod_{i=1}^k g(z^i(x^1,\ldots,x^k)).\]
\label{lem:prodlog}
\end{lemma}
\begin{proof}
This follows directly from the log-supermodularity of $g$.
\end{proof}

The proof of our variant of the ``$2k$ functions theorem'' uses the properties of weak majorizations:

\begin{definition}
A vector $x \in\mathbb{R}^n$ is \textbf{weakly majorized} by a vector $y\in \mathbb{R}^n$, denoted $x\prec_w y$, if
\[\sum_{i=1}^t z^i(x_1,\ldots,x_n) \leq \sum_{i=1}^t z^i(y_1,\ldots,y_n)\]
for all $t\in\{1,\ldots,n\}$.  
\end{definition}

For the purposes of this paper, we will only need the following result concerning weak majorizations:
\begin{theorem}
For $x,y\in\mathbb{R}^n$, $x\prec_w y$ if and only if 
\[\sum_{i=1}^n g(x_i) \leq \sum_{i=1}^n g(y_i)\]
for all continuous increasing convex functions $g:\mathbb{R}\rightarrow\mathbb{R}$.\label{thm:weaksum}
\end{theorem}
\begin{proof}
See 3.C.1.b and 4.B.2 of \cite{marshall}.
\end{proof}

We now state and prove our variant of the $2k$ functions theorem in two pieces.  First, we consider the case where $n=1$: 

\begin{lemma}
Let $f_1,\ldots,f_k:\{0,1\}\rightarrow \mathbb{R}_{\geq 0}$ and $g:\{0,1\}^{k}\rightarrow \mathbb{R}_{\geq 0}$  be nonnegative real-valued functions such that $g$ is log-supermodular.  If for all $x^1,\ldots,x^k\in\{0,1\}$,
\[g(x^1,\ldots,x^k) \leq \prod_{i=1}^k f_i(z^i(x^1,\ldots,x^k)),\]
then
\[ \sum_{x^1,\ldots,x^k} g(x^1,\ldots,x^k) \leq \prod_{i=1}^k \Big[\sum_{x\in\{0,1\}} f_i(x)\Big].\]
\label{lem:bcases}
\end{lemma}

\begin{proof}
Let $G\in\mathbb{R}^{2^k}$ be the vector whose $2^k$ elements correspond to the $2^k$ distinct evaluations of $g$.  Similarly, let $F\in\mathbb{R}^{2^k}$ be the vector whose $2^k$ elements correspond to the $2^k$ distinct evaluations of $f(x^1,\ldots,x^k) \triangleq \prod_{i=1}^k f_i(x^i)$.  Let $\log G \triangleq (\log G_1,\ldots, \log G_{2^k})$ and  $\log F \triangleq (\log F_1,\ldots, \log F_{2^k})$.  Our strategy will be to show that $\log G \prec_w \log F$.  Then, by Theorem \ref{thm:weaksum} and the fact that $2^x$ is convex and increasing, we will have
\begin{align*}
\sum_{x^1,\ldots,x^k} g(x^1,\ldots,x^k) &= \sum_{i=1}^{2^k} 2^{\log G_i}
\leq \sum_{i=1}^{2^k} 2^{\log F_i}
= \sum_{x^1,\ldots,x^k}\prod_{i=1}^k f_i(x^i)
\end{align*}
as desired.  We note that, by continuity arguments, this analysis holds even when some values of $g$ and $f$ are equal to zero.  Further, let $G^c\in\mathbb{R}^{k \choose c}$ be the vector obtained from $G$ by only considering assignments with exactly $c$ nonzero elements (i.e., $x^1+\ldots+x^k = c$), and define $F^c$ similarly for $F$.  If we can show that 
\[\prod_{m=1}^M z^m(G^c_1,\ldots,G^c_{k\choose c}) \leq \prod_{m=1}^M z^m(F^c_1,\ldots,F^c_{k\choose c})\]
for all $c\in\{0,\ldots,k\}$ and $M \leq {k \choose c}$, then we must have that $\log G \prec_w \log F$.  

Now, fix $c\in\{0,\ldots,k\}$, $T\in\{1,\ldots, {k \choose c}\}$, and let $V^c = \{ v\in\{0,1\}^k\hspace{.1cm}|\hspace{.1cm} v_1+\ldots+v_k = c\}$.  Suppose $v^1,\ldots, v^T\in V^c$ are $T$ distinct vectors.  By Lemma \ref{lem:prodlog}, we must have
\begin{align*}
\prod_{t=1}^T g(v^t)  &\leq \prod_{t=1}^T g(z^t(v^1,\ldots ,v^T))
\leq \prod_{t=1}^T f(w^t)
\end{align*}
where $w^t_j = z^j(z^t(v^1,\ldots,v^T)_1,\ldots,z^t(v^1,\ldots,v^T)_k)$ for each $j\in\{1,\ldots,k\}$.  Given any such $v^1,\ldots, v^T\in V^c$, we will show how to construct distinct vectors $\overline{v}^1,\ldots,\overline{v}^T\in V^c$ such that $\prod_{t=1}^T f(w^t) \leq \prod_{t=1}^T f(\overline{v}^t)$.  Consequently, we will have
\[\prod_{t=1}^T g(v^t)  \leq \prod_{t=1}^T f(\overline{v}^t) \leq \prod_{m=1}^T z^m(F^c_1,\ldots,F^c_{k\choose c}).\]
As our construction will work for any choice of distinct vectors $v^1,\ldots,v^T\in V^c$, it will work, in particular, for the $T$ distinct vectors in $V^c$ that maximize $\prod_{t=1}^T g(v^t)$, and the lemma will then follow as a consequence of our previous arguments. 

We now describe how to construct the vectors $\overline{v}^1,\ldots,\overline{v}^T$ from the vectors $v^1,\ldots, v^T$.  Let $A\in\mathbb{R}^{k\times t}$ be the matrix whose $i^{th}$ column is given by the vector $v^i$.  Construct $\overline{A}\in\mathbb{R}^{k\times t}$ from $A$ by swapping the rows of $A$ so that for each $i < j\in\{1,\ldots,k\}, \sum_p \overline{A}_{ip} \geq \sum_p \overline{A}_{jp}$.  Intuitively, the first row of $\overline{A}$ corresponds to the row of $A$ with the most nonzero elements, the second row of $\overline{A}$ corresponds to the row of $A$ with the second largest number of nonzero elements, and so on. Let $\overline{v}^1,\ldots,\overline{v}^T$ be the columns of $\overline{A}$.  Notice that $\overline{v}^1,\ldots,\overline{v}^T$ are distinct vectors in $V^c$ and that, by construction, $z^j(z^t(\overline{v}^1,\ldots,\overline{v}^T)_1,\ldots,z^t(\overline{v}^1,\ldots,\overline{v}^T)_k) = z^t(\overline{v}^1,\ldots,\overline{v}^T)_j$ for each $j\in\{1,\ldots,k\}$ and $t\in\{1,\ldots, T\}$. Therefore, we must have
\begin{align*}
\prod_{t=1}^T g(\overline{v}^t)  &\leq \prod_{t=1}^T g(z^t(\overline{v}^1,\ldots,\overline{v}^T))
\leq \prod_{t=1}^T f(z^t(\overline{v}^1,\ldots,\overline{v}^T))
=  \prod_{t=1}^T f(\overline{v}^t)
\end{align*}
where the equality follows from the definition of $f$ as a product of the $f_i$.  In addition, the vector $z^t(v^1,\ldots,v^T)$ is simply a permuted version of the vector $z^t(\overline{v}^1,\ldots,\overline{v}^T)$ which means that their $j^{th}$ largest elements must agree:
\begin{align*}
w^t_j &= z^j(z^t(v^1,\ldots,v^T)_1,\ldots,z^t(v^1,\ldots,v^T)_k)\\
&= z^j(z^t(\overline{v}^1,\ldots,\overline{v}^T)_1,\ldots,z^t(\overline{v}^1,\ldots,\overline{v}^T)_k)\\
&= z^t(\overline{v}^1,\ldots,\overline{v}^T)_j.
\end{align*}
Therefore,
\begin{align*}
\prod_{t=1}^T g(v^t) &\leq \prod_{t=1}^T f(w^t)
= \prod_{t=1}^T f(z^t(\overline{v}^1,\ldots,\overline{v}^T))
= \prod_{t=1}^T f(\overline{v}^t)
\end{align*}
and the lemma follows as a consequence .
\end{proof}

In the case that $n=1$ and $k\geq 1$, this lemma is a more general result than the $2k$ functions theorem:  if $g(x^1,\ldots,x^k) = \prod_i g_i(x^i)$ for $g_1,\ldots,g_k:\{0,1\}\rightarrow\mathbb{R}_{\geq 0}$, then $g$ is log-supermodular.  As in the proof of the 2k functions theorem, the general theorem for $n\geq 1$ follows by induction on $n$:
\begin{theorem}
Let $f_1,\ldots,f_k:\{0,1\}^n\rightarrow \mathbb{R}_{\geq 0}$ and $g:\{0,1\}^{kn}\rightarrow \mathbb{R}_{\geq 0}$  be nonnegative real-valued functions such that $g$ is log-supermodular.  If for all $x^1,\ldots,x^k\in\{0,1\}^n$,
\[g(x^1,\ldots,x^k) \leq \prod_{i=1}^k f_i(z^i(x^1,\ldots,x^k)),\]
then
\[ \sum_{x^1,\ldots,x^k} g(x^1,\ldots,x^k) \leq \prod_{i=1}^k \Big[\sum_{x\in\{0,1\}^n} f_i(x)\Big].\]\label{thm:kfunc}
\end{theorem}
\begin{proof}
We will prove the result for general $k$ and $n$ by induction on $n$. The base case of $n=1$ follows from Lemma \ref{lem:bcases}.  Now, for $n\geq 2$, suppose that the result holds for $k \geq 1$ and $n - 1$, and let $f_1,\ldots,f_k:\{0,1\}^n\rightarrow \mathbb{R}_{\geq 0}$ and $g:\{0,1\}^{kn}\rightarrow \mathbb{R}_{\geq 0}$ be nonnegative real-valued functions such that $g$ is log-supermodular. 

Define $f': \{0,1\}^{n-1}\rightarrow \mathbb{R}_{\geq 0}$ and $g':\{0,1\}^{k(n-1)}\rightarrow \mathbb{R}_{\geq 0}$ as 
\begin{align*}
f'_i(y) &= f_i(y,0) + f_i(y, 1)\\
g'(y^1,\ldots, y^k) &= \sum_{s^1,\ldots,s^k\in\{0,1\}} g(y^1, s^1,\ldots, y^k, s^k)
\end{align*}
Notice that $g'$ is log-supermodular because it is the marginal of a log-supermodular function (see Lemma \ref{lem:logsup}).  If we can show that
\begin{align*}
g'(y^1,\ldots,y^k) \leq \prod_{i=1}^k f'_i(z^i(y^1,\ldots,y^{k}))
\end{align*}
for all $y^1,\ldots,y^k\in\{0,1\}^{n-1}$, then the result will follow by induction on $n$.  To show this, fix $\overline{y}^1,\ldots,\overline{y}^k\in\{0,1\}^{n-1}$ and define $\overline{f}: \{0,1\}\rightarrow \mathbb{R}_{\geq 0}$ and $\overline{g}:\{0,1\}^{k}\rightarrow \mathbb{R}_{\geq 0}$ as 
\begin{align*}
\overline{f}_i(s) &= f_i(z^i(\overline{y}^1,\ldots,\overline{y}^k), s)\\
\overline{g}(s^1,\ldots,s^k) &= g(\overline{y}^1,s^1,\ldots,\overline{y}^k,s^k)
\end{align*}

We can easily check that $\overline{g}(s^1,\ldots ,s^k)$ is log-supermodular and that $\overline{g}(s^1,\ldots,s^k) \leq \prod_{i=1}^k \overline{f}_i(z^i(s^1,\ldots,s^k))$ for all $s^1,\ldots, s^k\in\{0,1\}$.  Hence, by Lemma \ref{lem:bcases},
\begin{align*}
g'(\overline{y}^1, \ldots, \overline{y}^k) &= \sum_{s^1,\ldots,s^k} \overline{g}(s^1,\ldots,s^k)
\leq \prod_{i=1}^k \sum_{s\in\{0,1\}} \overline{f}_i(s)
= \prod_{i=1}^k f'_i(z^i(\overline{y}^1,\ldots,\overline{y}^k))
\end{align*}
which completes the proof of the theorem.
\end{proof}

\section{Graph Covers and the Partition Function}
In this section,  we show how to apply Theorem \ref{thm:kfunc} in order to resolve Conjecture \ref{conj:bethecover}.  In addition, we show that the theorem can be applied, more generally, to yield similar results for a class of functions that can be converted into a log-supermodular functions by a change of variables.

\subsection{Log-supermodularity and Graph Covers}
\label{sec:part}
The following theorem follows easily from Theorem \ref{thm:kfunc}:

\begin{theorem}
If $f^G:\{0,1\}^n\rightarrow \mathbb{R}_{\geq 0}$ admits a log-supermodular factorization over $G=(V,\mathcal{A})$, then for any $k$-cover, $H$, of $G$, $Z(H) \leq Z(G)^k$. \label{thm:kcover}
\end{theorem}
\begin{proof}
Let $H$ be a $k$-cover of $G$.  Divide the vertices of $H$ into $k$ sets $S_1,\ldots,S_k$ such that each set contains exactly one copy of each vertex $i\in V$.  Let the assignments to the variables in the set $S_i$ be denoted by the vector $x^i$.  

For each $\alpha\in\mathcal{A}$, let $y^i_\alpha$ denote the assignment to the $i^{th}$ copy of $\alpha$ by the elements of $x^1,\ldots,x^k$.  By Lemma \ref{lem:prodlog},
\begin{align*}
\prod_{i=1}^k \psi_\alpha(y^i_\alpha) &\leq \prod_{i=1}^k \psi_\alpha(z^i(y^1_\alpha,\ldots,y^k_\alpha))
= \prod_{i=1}^k \psi_\alpha(z^i(x^1_\alpha,\ldots,x^k_\alpha))
= \prod_{i=1}^k \psi_\alpha(z^i(x^1,\ldots,x^k)_\alpha)
\end{align*}

From this, we can conclude that $f^H(x^1,\ldots,x^k) \leq \prod_{i=1}^k f^G(z^i(x^1,\ldots,x^k))$.  Now, by Theorem \ref{thm:kfunc},
\[Z(H) = \sum_{x^1,\ldots,x^k} f^H(x^1,\ldots,x^k) \leq \prod_{i=1}^k [\sum_{x^i} f^G(x^i)] = Z(G)^k\]
\end{proof}

This theorem settles the conjecture of \cite{sudderth07} for any log-supermodular function that admits a pairwise binary factorization.  Indeed, the above theorem solves the problem for a larger class of log-supermodular graphical models:
\begin{corollary}
If $f:\{0,1\}^n\rightarrow \mathbb{R}_{\geq 0}$ admits a log-supermodular factorization over $G = (V,\mathcal{A})$, then $Z_\mathrm{B}(G) \leq Z(G)$.
\end{corollary}
\begin{proof}
This follows directly from Theorem \ref{thm:kcover} and Theorem \ref{thm:pascal}.
\end{proof}
As the value of the Bethe approximation at any of the fixed points of BP is always a lower bound on $Z_\mathrm{B}(G)$, the conclusion of the corollary holds for any fixed point of the BP algorithm as well.

\begin{corollary}
If $f:\{0,1\}^n\rightarrow \mathbb{R}_{\geq 0}$ admits a log-supermodular factorization over $G = (V,\mathcal{A})$, then 
\[Z_\mathrm{B}(G) = \lim_{k\rightarrow\infty} \sqrt[k]{\sum_{H\in \mathcal{C}^k(G)} Z(H)/|\mathcal{C}^k(G)|}\]
where $\mathcal{C}^k(G)$ is the set of all $k$-covers of $G$. 
\end{corollary}
\begin{proof}
By Theorem \ref{thm:kcover} and the definition of $Z_\mathrm{B}$, $Z(H) \geq Z_\mathrm{B}(H)\geq Z_\mathrm{B}(G)^k$ for any k-cover $H$ of $G$.  The corollary then follows from Theorem \ref{thm:pascal}.
\end{proof}

%

\subsection{Beyond Log-supermodularity} \label{sec:switching}
While Theorem \ref{thm:kcover} is a statement only about log-supermodular functions, we can use Theorem \ref{thm:kfunc} to infer similar results even when the function under consideration is not log-supermodular.  As an example of such an application, we consider the problem of counting the number of independent sets in a given graph, $G = (V,E)$.  An independent set, $I\subseteq V$, in $G$ is a subset of the vertices such that no two adjacent vertices are in $I$.  We define the following function:
\[I^G(x_1,\ldots,x_{|V|}) = \prod_{(i,j)\in E} (1-x_ix_j)\]
which is equal to one if the nonzero $x_i$'s define an independent set and zero otherwise.  As every potential function depends on at most two variables, $I^G$ factorizes over the graph $G=(V,E)$.  Notice that $f^G$ is log-submodular, not log-supermodular.

In this section, we will focus on bipartite graphs:  $G = (V, E)$ is bipartite if we can partition the vertex set into two sets $A\subseteq V$ and $B= V\setminus A$ such that $A$ and $B$ are independent sets.  Examples of bipartite graphs include single cycles, trees, and grid graphs.  We will denote bipartite graphs as $G = (A,B,E)$.

For any bipartite graph $G = (A,B,E)$, $I^G$ can be converted into a log-supermodular graphical model by a simple change of variables.  Define $y_a = x_a$ for all $a\in A$ and $y_b = 1 - x_b$ for all $b\in B$.  We then have
\begin{align*}
I^G(x_1,\ldots,x_{|V|}) & =  \prod_{(i,j)\in E} (1-x_ix_j)\\
& =  \prod_{(a,b)\in E, a\in A, b\in B} (1-y_a(1-y_b))\\
& \triangleq \overline{I}^G(y_1,\ldots,y_{|V|}).
\end{align*}
$\overline{I}^G$ admits a log-supermodular factorization over $G$ and $\sum_{y} \overline{I}^G(y) = \sum_x I^G(x)$.  Similarly, for any graph cover $H$ of $G$, we have $\sum_{y} \overline{I}^H(y) = \sum_x I^H(x)$.  Consequently, by Theorem \ref{thm:kfunc}, we can conclude that $Z(G) \geq Z_\mathrm{B}(G)$.  

Similar observations can, for example, be used to show that the Bethe partition function provides a lower bound on the true partition function for other problems that factor over pairwise bipartite graphical models (e.g., the antiferromagnetic Ising model on a grid, counting the number of vertex covers of a bipartite graph, counting the number of satisfying assignments of a monotone 2-SAT instance whose corresponding graphical structure is bipartite).

\section{Conclusions}

While the results presented above were discussed in the case that the temperature parameter, $T$, was equal to one, they easily extend to any $T\geq 0$ (as exponentiation preserves log-supermodularity in this case).  Hence, all of the bounds discussed above can be extended to the problem of maximizing a log-supermodular function.  In particular, the inequality in Theorem \ref{thm:kcover} suggests that the maximizing assignment on any graph cover must correspond to a lift of a maximizing assignment on the base graph.

This work also suggests a number of directions for future research.  While the above work provides lower bounds on the partition function, similar ideas may be able to provide upper bounds as well.  We note that related work on the Bethe approximation for permanents has already begun to explore these possibilities \cite{betheperm}\cite{gurvitsnew}.  Similarly, an analog of Theorem \ref{thm:kfunc} for log-submodular functions may also be useful in the pursuit of upper bounds. The primary difficulty is that marginal distributions of log-submodular functions are not necessarily log-submodular, but perhaps upper bounds can be obtained when restricting to families of log-submodular functions all of whose marginals are also log-submodular.

\section*{Acknowledgments}
The author would like to thank Nicolas Macris, for many useful discussions about the ferromagnetic Ising model and correlation inequalities, and Pascal Vontobel, for his comments and suggestions during the preparation of this work.

\bibliographystyle{ieeetr}
\bibliography{biblio}

\end{document}